\documentclass[10pt, conference]{IEEEtran}

\title{GULPS: Two-Qubit Gate Synthesis via Linear Programming for Heterogeneous Instruction Sets}

\author{
    \IEEEauthorblockN{Evan McKinney\,\orcidlink{0000-0002-4865-5458}\IEEEauthorrefmark{1}\IEEEauthorrefmark{2}\IEEEauthorrefmark{3}\IEEEauthorrefmark{4},  Lev S. Bishop\,\orcidlink{0000-0003-1318-1149}\IEEEauthorrefmark{4}}
    \IEEEauthorblockA{
    \IEEEauthorrefmark{1}Department of Computer Science, Yale University, New Haven, CT, USA \\
    \IEEEauthorrefmark{2}Yale Quantum Institute, Yale University, New Haven, CT, USA \\
    \IEEEauthorrefmark{3}Department of Electrical and Computer Engineering, University of Pittsburgh, Pittsburgh, PA, USA \\
    \IEEEauthorrefmark{4}IBM Quantum, IBM T.J. Watson Research Center, Yorktown Heights, NY, USA \\
    evan.mckinney@yale.edu, lsbishop@us.ibm.com}
}

\usepackage[sort, compress]{cite}
\usepackage{amsmath,amssymb,amsfonts}
\usepackage{graphicx}
\usepackage{textcomp}
\usepackage{xcolor}
\usepackage[hyphens]{url}
\usepackage{fancyhdr}
\usepackage{hyperref}
\usepackage{amsthm}
\usepackage{svg}
\usepackage{tikz}
\usepackage{orcidlink}
\usetikzlibrary{quantikz2}
\usepackage{listings}
\usepackage{booktabs}

\newtheorem{theorem}{Theorem}[section]

\usepackage{xcolor}
\usepackage{ifthen}
\usepackage{mathabx}
\newboolean{showcomments}
\setboolean{showcomments}{true}
\ifthenelse{\boolean{showcomments}}
{ \newcommand{\mynote}[3]{
     \fbox{\bfseries\sffamily\scriptsize#1}
        {\small$\blacktriangleright$\textsf{\emph{\color{#3}{#2}}}$\blacktriangleleft$}}}
{ \newcommand{\mynote}[3]{}}

\begin{document}

\maketitle

\thispagestyle{plain}
\pagestyle{plain}

\begin{abstract}
Modern quantum hardware exposes heterogeneous two-qubit instruction sets through fractional, continuously parameterized, and per-pair native gates, but synthesis remains largely framed around CNOT and a small catalog of closed-form rules. We present \textbf{GULPS} (Global Unitary Linear Programming Synthesis), a two-qubit compiler that partitions synthesis into depth-$2$ segments and uses a linear program over quantum Littlewood--Richardson reachability inequalities to plant the intermediate invariants between them. Each segment becomes an independent low-dimensional least-squares fit, solved by a Gauss--Newton/Levenberg--Marquardt routine. On Haar-random two-qubit targets, GULPS is more than $500{\times}$ faster than the general-purpose synthesizers BQSKit and NuOp at strictly lower circuit cost. Against Qiskit's specialized \texttt{XXDecomposer} on $XX$-family ISAs, GULPS produces identical output circuits $3.9$--$9.2{\times}$ faster, compounding to $7$--$19{\times}$ on full-circuit transpilation. All decompositions reach the double-precision unitary-infidelity floor. As a byproduct, the continuous formulation yields a Haar-averaged lower bound on expected circuit cost, against which discrete calibration choices can be benchmarked. GULPS is distributed on PyPI and registers as a Qiskit translation-stage plugin.
\end{abstract}

\begin{IEEEkeywords}
Quantum computing, quantum gates, quantum circuit, compilers, linear programming
\end{IEEEkeywords}


\section{Introduction}

\begin{figure}[t]
    \centering
    \includegraphics[width=\columnwidth]{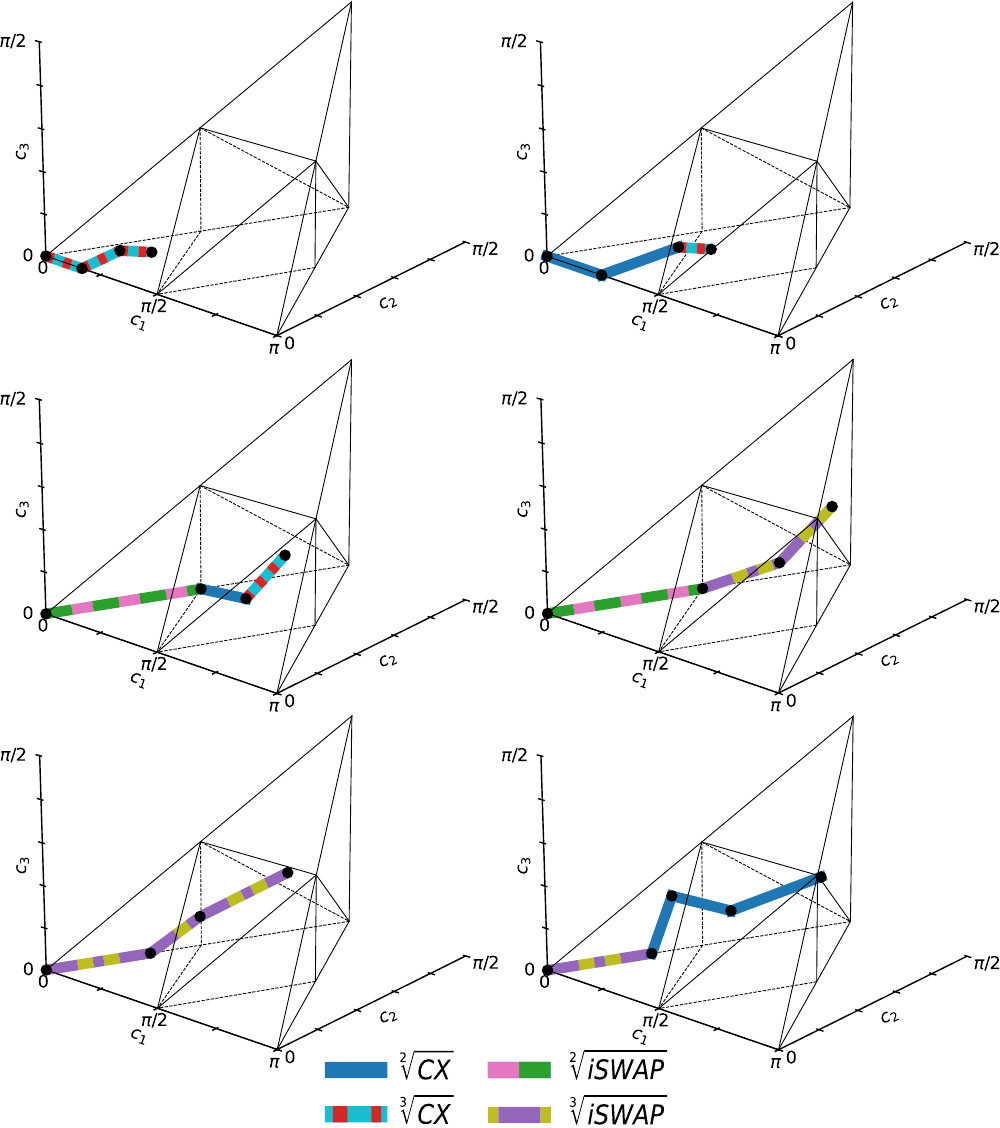}
    \caption{Segmented Weyl-chamber trajectories produced by GULPS on the heterogeneous ISA \(\mathcal{G} = \{\sqrt{\text{CX}},\sqrt[3]{\text{CX}},\sqrt{\text{iSWAP}},\sqrt[3]{\text{iSWAP}}\}\). Each panel is one Haar-random target. Markers trace the canonical-invariant point $(c_1, c_2, c_3)$ after each basis-gate segment, connected by depth-2 hops; chains begin at the identity and terminate at the target.}
    \label{fig:decomps}
\end{figure}

Two-qubit interactions are the primary nonlocal resource in quantum circuits. Beyond generating entanglement, they dominate the cost of implementing algorithmic primitives, decomposing larger unitaries, and routing circuits across hardware with limited connectivity. Variational algorithms use them to simulate two-body Hamiltonian terms; algorithms such as Shor's and Grover's compile through generic decompositions (Cosine--Sine, Quantum Shannon) into sequences of controlled-phase operations~\cite{shende2005synthesis, drury2008constructive, saeedi2010block, wierichs2025recursive}.

Recent hardware advances are producing increasingly heterogeneous two-qubit instruction set architectures (ISAs). Frame-tracking control techniques convert a nominal entangler such as cross-resonance into a broader family of SU(4) operations~\cite{wei2024native, chen2025efficient}, and fractional~\cite{perez_error-divisible_2023, almeida_fractional_2024} or continuously parameterized~\cite{foxen2020demonstrating, scarato2025realizing, sugawara20254, yale2025realization, yale2024noise, chen2024one} pulses expose entire one-parameter families of distinct two-qubit primitives. \textit{Devices are thus better viewed as exposing heterogeneous entangling libraries than a single native gate.}

Despite this hardware diversity, two-qubit synthesis remains largely framed around CNOT~\cite{vatan_optimal_2004, madden_best_2022}. Closed-form synthesis rules exist for a few special families (B gate, XX-family gates, iSWAP, $\sqrt{\text{iSWAP}}$), but none operate over a mixed library. A compiler that mixes entangling primitives can lower synthesis cost, expand the design space for parameterized ansätze, and enable routing strategies unavailable when synthesis is constrained to a single basis~\cite{lao_designing_2021, javadi_improving_2023, mckinney2024mirage, kalloor2024quantum}. Targeting native unitaries directly also avoids the pulse-level overhead of refocusing unwanted terms to emulate a textbook gate. Recent work calibrates each qubit pair to its own native entangling basis instead of enforcing a CNOT abstraction~\cite{lin2022let, kakkar2025no}.

Existing numerical methods can handle arbitrary two-qubit bases, but scale poorly as the ISA grows richer: runtime increases rapidly, and convergence degrades as each candidate sentence induces a new optimization task~\cite{rakyta2022approaching, younis_qfast_2020, davis2020towards, smith_leap_2022, nemkov_efficient_2023}.

\textbf{GULPS} is a two-qubit compiler for heterogeneous ISAs. It views the target unitary as a trajectory through gate-invariant space, from identity to the target. The picture is motivated by continuous Cartan flows generated by time-dependent control Hamiltonians~\cite{zhang_geometric_2003}. Rather than solving for the entire decomposition at once, it partitions the trajectory into depth-two segments and \textbf{uses a linear program over quantum Littlewood--Richardson (QLR) reachability inequalities to plant the intermediate invariants between them.} Each segment then becomes an independent low-dimensional fit, recovered by a lightweight numerical solver. The LP plays a dual role. It certifies that the chosen sentence can reach the target, and supplies the intermediates that decouple synthesis into per-segment subproblems.

This formulation leads to three contributions:
\begin{itemize}
\item \textbf{Sentence selection with intermediates:} A cost-ordered search over staircase monodromy constraints identifies the cheapest feasible gate sequence and emits the intermediate invariants between segments.
\item \textbf{Fixed-size local recovery:} A Gauss--Newton/Levenberg--Marquardt routine with analytic Jacobians solves each segment as an independent fit whose dimension does not grow with circuit depth.
\item \textbf{Empirical results:} GULPS is more than $500{\times}$ faster than BQSKit and more than $2500{\times}$ faster than NuOp on heterogeneous iSWAP-power ISAs, and $3.9$--$9.2\times$ faster than Qiskit's \texttt{XXDecomposer} on XX-family ISAs, while matching or improving circuit cost and reaching the double-precision infidelity floor.
\item \textbf{Calibration baseline:} A continuous-ISA formulation provides a Haar-averaged lower bound on expected circuit cost. On standard base-gate families, two well-chosen fractional strengths recover most of the achievable benefit, suggesting a small calibration budget already captures the value of a fully tunable family.
\end{itemize}

The remainder of the paper covers background on two-qubit invariants and monodromy polytopes (\S\ref{sec:background}), the segmented LP method (\S\ref{sec:methods}), benchmarks against NuOp, BQSKit, and Qiskit's \texttt{XXDecomposer} (\S\ref{sec:evaluation}), a discussion of timing, tail behavior, and scope (\S\ref{sec:discussion}), and a treatment of continuous gate families as a calibration-design reference (\S\ref{sec:continuous}).

\section{Background}
\label{sec:background}

\subsection{Two-qubit invariants and Cartan decomposition}

Any two-qubit unitary \( U \in \mathrm{PU}(4) \) admits a Cartan (KAK) decomposition~\cite{tucci_introduction_2005} that separates nonlocal interaction from local single-qubit rotations:
\begin{equation}
    \label{eq:canonical-decomp}
    U = K \cdot \mathrm{CAN}(a_1,a_2,a_3) \cdot K',
\end{equation}
\begin{equation*}
    \begin{tikzcd}
    & \gate[2]{U} & \qw \\
    & \qw         & \qw
    \end{tikzcd}
    =
    \begin{tikzcd}
    & \gate{K'} & \gate[2]{\mathrm{CAN}(a_1,a_2,a_3)} & \gate{K} & \qw \\
    & \gate{K'} & \qw                           & \gate{K} & \qw
    \end{tikzcd}
\end{equation*}
where \(K,\,K'\in\mathrm{PU}(2)\times\mathrm{PU}(2)\) are local unitaries, and
\begin{equation}
    \mathrm{CAN}(a_1,a_2,a_3) = \exp\Bigl(-i\bigl(a_1 XX + a_2 YY + a_3 ZZ\bigr)\Bigr)
    \label{eq:CAN}
\end{equation}
is the canonical two-qubit gate. When restricted to the Weyl chamber, the parameters \((a_1,a_2,a_3)\) uniquely characterize \(U\) up to local equivalence~\cite{zhang_geometric_2003, tucci_introduction_2005, makhlin2002nonlocal}.

Local equivalence $U \simeq V$ holds when $U$ and $V$ differ only by local pre- and post-multiplication, and is fully characterized by matching Cartan coordinates:
\[
\mathrm{CAN}(U)=\mathrm{CAN}(V)\;\Longleftrightarrow\; U = K\,V\,K'
\]
for some $K, K'\in\mathrm{PU}(2)\times\mathrm{PU}(2)$. In synthesis, surrounding single-qubit layers absorb local discrepancies between intermediate segments.

Because single-qubit rotations are typically inexpensive (e.g., virtual-$z$ rotations)~\cite{mckay_efficient_2017}, synthesis focuses on nonlocal invariants. We use three equivalent~\cite{watts2013metric} representations, each suited to a different role:

\textbf{Weyl coordinates.} The Cartan parameters \((a_1,a_2,a_3)\) specify a unique point in the Weyl chamber, giving the primary geometric description of two-qubit nonlocality, and we use them for visualization and KAK decompositions.

\textbf{Makhlin invariants.} Let \(U^Q := Q^\dagger U Q\) denote the magic-basis representation and define
\begin{equation}
    m(U) := (U^Q)^T U^Q.
\end{equation}
In the magic basis, local unitaries act as $\mathrm{SO}(4)\times\mathrm{SO}(4)$, so symmetric functions of the eigenvalues of $m(U)$ are two-sided local invariants. The pair $(g_1, g_2)$, normalized by $\det U$ to absorb global phase, separates local-equivalence classes~\cite{makhlin2002nonlocal}:
\begin{equation}
    g_1(U) := \frac{(\mathrm{tr}\,m(U))^2}{16\,\det U}, \qquad
    g_2(U) := \frac{(\mathrm{tr}\,m(U))^2 - \mathrm{tr}(m(U)^2)}{4\,\det U}.
\end{equation}
We use $(g_1, g_2)$ to cast the per-segment local-recovery step as a low-dimensional least-squares problem (\S\ref{sec:methods}).

\textbf{Monodromy coordinates.} Define
\begin{equation}
    \gamma_Q(U) := U(\sigma_y \otimes \sigma_y)U^T(\sigma_y \otimes \sigma_y),
\end{equation}
with eigenvalues \(e^{2\pi i \lambda_j}\). The ordered log-spectrum \(\lambda\), mapped into a fixed Weyl alcove, yields additive coordinates in which composition constraints become linear inequalities~\cite{peterson2020fixed}, and we use them to construct LP feasibility conditions in synthesis.

\subsection{Heterogeneous ISAs}
\label{sec:heterogeneous_isas}

Two-qubit interactions on modern hardware arise from tunable Hamiltonians whose effective form depends on control amplitudes, detunings, and interaction time, so platforms naturally realize families of entanglers rather than a single primitive. In superconducting systems with tunable couplers, this includes an iSWAP-family (mediated by microwave photon exchange) and CZ (phase accumulation on the $|f\rangle$ state)~\cite{krantz2019quantum}. In contrast, certain bosonic encodings instead yield families of logical ZZ and SWAP-family interactions~\cite{tsunoda2023error}. Analogous families appear, for example, in fluxonium charge coupling, quantum-dot exchange, M{\o}lmer--S{\o}rensen gates in trapped ions, and Rydberg blockade in neutral atoms.

Compilation interfaces expose these capabilities as a discrete collection of calibrated gates sampled from one or more continuous evolutions. We call such a library a \emph{heterogeneous instruction set architecture (ISA)}: a finite set $\mathcal{G}=\{G^{(1)},\ldots,G^{(m)}\}\subset \mathrm{PU}(4)$ of non-locally inequivalent primitives with different implementation costs. Heterogeneity arises either as (i) \emph{fractional} gates truncating one evolution at different times, or (ii) \emph{qualitatively distinct} entanglers from different Hamiltonians.

We adopt an additive cost model: each gate $G\in\mathcal{G}$ has weight $c(G)\ge 0$ (e.g., duration), and circuit cost is the sum. Each two-qubit segment is bracketed by a single-qubit layer of nonzero cost. Two-qubit gate count and two-qubit depth coincide in this setting (one 2q gate per layer). The single-qubit-layer cost adds a per-segment offset that favors fewer-but-stronger gates over longer chains of weaker ones, for example two cost-$1/2$ gates over four cost-$1/4$ gates of equal two-qubit weight. Minimizing total cost therefore subsumes both standard metrics. As the next subsection makes precise, compilation requires selecting a sequence of basis gates whose reachable invariant region contains the target, then recovering compatible local unitaries.


\subsection{Circuit sentences and monodromy polytope inequalities}

A \emph{circuit sentence} is a fixed sequence of two-qubit basis gates from \(\mathcal{G}\) interleaved with local unitaries:
\begin{equation}
S : \theta \mapsto \Bigl( K_0(\theta),\, G_1,\, K_1(\theta),\, G_2,\, \dots,\, G_n,\, K_n(\theta) \Bigr),
\end{equation}
where each \(G_j \in \mathcal{G}\) is fixed and each \(K_j(\theta)\in\mathrm{PU}(2)\times\mathrm{PU}(2)\) varies. Sweeping the local rotations carves out a convex region in invariant space, the \emph{circuit polytope}, which depends only on \((G_1,\dots,G_n)\). The polytope is invariant under reordering of the basis gates (Appendix~\ref{app:sentence_ordering}).

\textit{Whether a given sentence can realize a target reduces to testing containment of the target's invariants in the sentence's circuit polytope.} We build up to this depth-$n$ object in two stages. The atomic case is depth-$2$, where reachability admits an exact characterization as finite linear inequalities in monodromy coordinates. This is a special case of the more general \emph{monodromy polytope}~\cite{agnihotri1998eigenvalues, belkale2001local}, here applied to two-gate composition. The general circuit polytope then follows by chaining these depth-$2$ slices through intermediate invariants.

Concretely, depth-2 reachability is a \emph{multiplicative eigenvalue problem}: for two fixed gates \(G_1, G_2\), the canonical parameters of any composed unitary
\[
T = K_2\, G_2\, K_1\, G_1\, K_0
\]
are constrained by linear inequalities indexed by quantum Littlewood--Richardson (QLR) triples (Theorem 23 of~\cite{peterson2020fixed}). Fix integers \(r,k>0\) with \(r+k=4\), and let $Q_{r,k}$ denote the strictly increasing length-$r$ sequences with entries in $\{1,\ldots,r+k\}$. For each triple $a, b, c \in Q_{r,k}$ and integer level $d$ such that the QLR coefficient $N^{\,c,d}_{a,b}(r,k) = 1$, the parameters satisfy
\begin{equation}
\label{eq:qlr_inequalities}
d - \sum_{i=1}^{r} \Bigl(\alpha_{k+i} - a_i\Bigr) - \sum_{i=1}^{r} \Bigl(\beta_{k+i} - b_i\Bigr) + \sum_{i=1}^{r} \Bigl(\delta_{k+i} - c_i\Bigr) \ge 0.
\end{equation}
Here $\alpha, \beta, \delta$ are the sorted log-spectrum coordinates of $G_1$, $G_2$, $T$ in monodromy form (see~\cite{peterson2020fixed} for the full inequality table). Fixed-depth synthesis reduces to a linear feasibility check.

\subsection{Limitations of global numeric synthesis}

\begin{figure}[t]
    \centering
    \includegraphics[width=0.99\columnwidth]{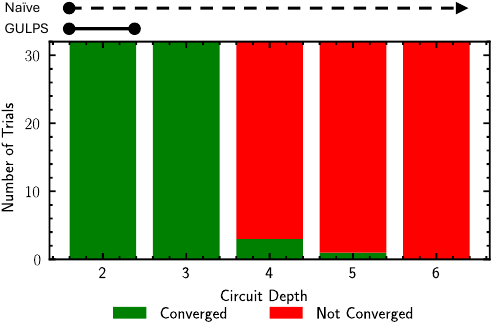}
    \caption{
    Convergence of the joint-circuit ansatz $G\,(R\otimes R)\,G\,\cdots\,G$ on the apex (largest-$c_3$ vertex) of the depth-$d$ polytope under \(\mathcal{G}=\{\sqrt[4]{\text{iSWAP}}\}\). Each trial optimizes $6(d-1)$ parameters jointly against a 3-component Makhlin residual; 32 independent Levenberg--Marquardt restarts per depth, convergence threshold $10^{-8}$ per component.
    }
    \label{fig:depth_convergence_failure}
\end{figure}

Without LP-supplied intermediate Weyl coordinates, a naive synthesizer falls into one of two failure modes. \emph{Joint optimization} fits a depth-$d$ parameterized ansatz end-to-end against the target's invariants: $6(d-1)$ local-rotation parameters against a 3-component residual, underdetermined and riddled with local minima, with parameter count growing in depth. \emph{Greedy segment-wise synthesis} converges one segment at a time without knowing the correct intermediate point, picking whatever drives the trailing partial product closer to the target. This trades the parameter blow-up for a myopic objective. The greedy intermediate need not lie on any feasible continuation, so the next segment may have no solution at all.

To isolate the joint failure mode, we run a convergence experiment on \(\mathcal{G}=\{\sqrt[4]{\text{iSWAP}}\}\). For each depth $d \in \{2,\dots,6\}$ we take the most extreme reachable point of the depth-$d$ polytope (largest-$c_3$ vertex) as a hard target, and run 32 independent Levenberg--Marquardt restarts on the joint Makhlin residual, counting convergence below $10^{-8}$ per component (max $2048$ function evaluations per restart). This setup resembles the inner loop of NuOp~\cite{lao_designing_2021}.

Fig.~\ref{fig:depth_convergence_failure} shows convergence collapsing abruptly with depth. The targets remain reachable at every depth tested, so feasibility is not the obstruction. There are simply too many free parameters at once: $d-1$ interleaved single-qubit layers contribute $6(d-1)$ angles, all solved for together, which random-restart Levenberg--Marquardt cannot navigate as $d$ grows. \textbf{The next section develops the LP-supplied intermediates that decouple this global non-convex search into independent low-dimensional subproblems.}

\section{Segmented Cartan trajectories}
\label{sec:methods}

\begin{figure}[t]
    \centering
    \includegraphics[width=0.9\columnwidth]{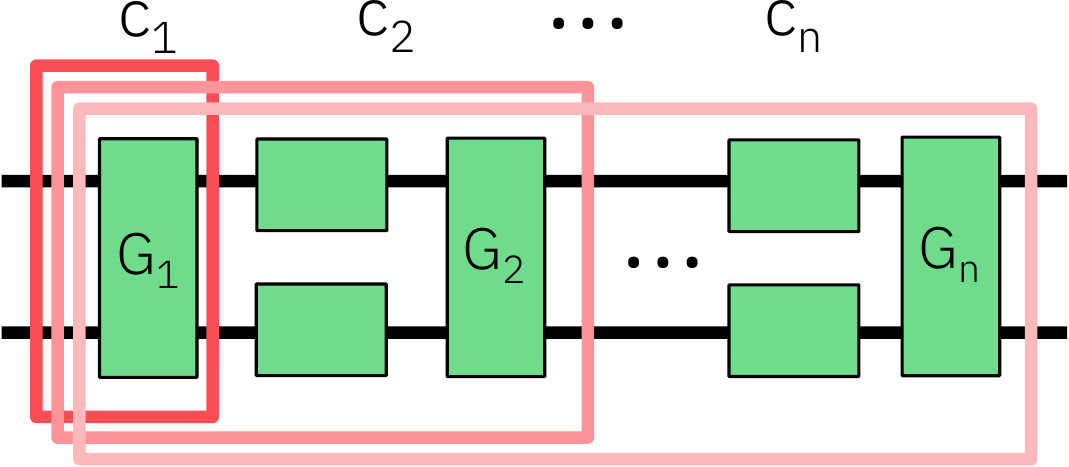}
    \caption{Schematic depth-$n$ decomposition circuit. Basis gates $G_i$ separate pairs of local rotations, with the canonical invariant $C_i$ recorded after each segment.}
    \label{fig:can}
\end{figure}

GULPS decomposes synthesis into a sequence of invariant-space segments, where each basis gate induces a constrained transition between consecutive invariant points. Let \(C_i \in \mathbb{R}^3\) denote the invariant after step \(i\). Each segment is governed by a QLR feasibility constraint of the form
\begin{equation}
\mathcal{L}(C_{i-1}, G_i, C_i) \le 0,\quad i=1,\dots,n,
\end{equation}
where \(G_i\) is the \(i\)-th basis gate (Fig.~\ref{fig:can}). In the notation of \S\ref{sec:background}, the log-spectrum coordinates $\alpha, \beta, \delta$ of two composing gates and their product specialize to $C_{i-1}, G_i, C_i$ in the per-segment formulation.

A valid decomposition requires simultaneous satisfaction of all segment constraints:
\begin{equation}
\bigcap_{i=1}^{n} \left\{\mathcal{L}(C_{i-1}, G_i, C_i) \le 0 \right\},
\label{eq:inequalities_union_fixed}
\end{equation}
subject to boundary conditions
\begin{equation}
\mathrm{CAN}(C_0) \simeq I, \qquad \mathrm{CAN}(C_n) \simeq T.
\end{equation}

Because \(C_1\) is fully determined by \(G_1\), the first constraint block can be absorbed into the second, yielding an equivalent formulation in which only intermediate invariants \(C_2,\dots,C_{n-1}\) remain free. Consequently, an \(n\)-gate sentence yields \(3(n-2)\) optimization variables across \(72(n-1)\) linear inequalities.


\subsection{Linear programming formalism}

These constraints assemble into a linear program (LP) over the intermediate invariants $C_1, \dots, C_{n-1}$, and in some cases the gate parameters as well. The form of $x$ depends on whether the ISA is discrete or continuous.

\paragraph{Discrete ISA, fixed sentence}
When each $G_i$ is drawn from a discrete set $\mathcal{G}$ and the sentence $(G_1, \dots, G_n)$ is fixed in advance, the gate monodromies are constants. Only the intermediate invariants are unknown:
\begin{equation}
x = [C_2, C_3, \dots, C_{n-1}] \in \mathbb{R}^{3(n-2)}, \qquad \min_x c^\top x \;\;\text{s.t.}\;\; A x \le b,
\end{equation}
with the QLR inequalities stacked into $A$, the fixed gate and target invariants absorbed into $b$, and an objective vector $c$ that regularizes the feasible point (set below). The LP returns a chain of intermediate invariants $(C_2, \dots, C_{n-1})$ that the downstream numerical stage uses as segment targets.

\paragraph{Continuous ISA, single family}
In a continuously parameterized family, gates are indexed by a strength $k \in [0,1]$, with $k = 0$ the identity and $k = 1$ the full base gate. For the families we treat here, the monodromy at strength $k$ is linear in $k$: if $B$ is the monodromy at $k = 1$, then strength $k$ has monodromy $B k$. This linearity holds for time-independent single-Cartan generators evolved for variable time, covering the iSWAP, XX, and fSim-style families considered in this paper; explicitly time-dependent control envelopes are left to future work. Substituting $g_i = B k_i$ for the gate monodromy in each QLR block keeps every constraint linear in the new unknowns, so the LP gains the strengths $k_i$ as decision variables alongside the intermediate invariants.

Two structural features promote this LP to a mixed-integer program. Each $k_i$ is \emph{semicontinuous}: a slot is either empty ($k_i = 0$) or active above the calibration floor ($k_i \in [k_{\min}, 1]$), encoded by a binary indicator $z_i \in \{0,1\}$ with the linking constraint $k_{\min}\, z_i \le k_i \le z_i$. Monotonic ordering of active strengths, $k_1 \ge k_2 \ge \cdots$, rules out permutations of the same multiset. Section~\ref{sec:continuous} uses this formulation as a Haar-weighted lower bound on the cost achievable by any discrete sampling of $G^k$.

In both formulations, the QLR inequalities couple only adjacent invariants, so the constraint matrix has a \emph{staircase} (block-tridiagonal) structure~\cite{fourer1982solving}. Two consequences shape the solver design. First, $A$ depends only on the length of the sentence; gate identities and the target enter only through $b$. A single dual-simplex solver therefore handles every length-$n$ instance, regardless of which gates appear or what the target is. Each new target perturbs only the last few entries of $b$, so a warm-started basis re-converges in a handful of pivots rather than a cold solve from scratch. Second, the local coupling makes infeasibility propagate fast. A violation in an early segment restricts every downstream feasible region, so the simplex typically rejects an infeasible sentence within a few iterations. 

We therefore set $c = -\mathbf{1}$, giving the simplex a direction to pivot in and pulling the intermediate invariants toward the polytope interior, away from the degenerate Weyl faces that complicate downstream local-rotation recovery. This also stabilizes the answer: instead of returning \emph{some} feasible point, the LP returns \emph{the} point that maximizes $\sum C_i$. Two similar targets land on similar intermediates rather than scattering across the feasible region, so a cached basis re-converges in a handful of pivots across clusters of related unitaries.

For discrete ISAs of moderate size ($|\mathcal{G}| \lessapprox 10$), GULPS enumerates sentences lazily through a cost-ordered priority queue, admitting each only when its gate costs are non-decreasing along the sentence, and testing LP feasibility in turn until the first accept. Constraint matrices are cached for reuse on later targets. As an alternative for batched compilation, the \texttt{monodromy} package can precompute the union of cost-ordered sentence polytopes.

\subsection{Numerical synthesis per segment}
\label{sec:numerics}

At this point the LP has fixed both the basis-gate sequence $\{G_i\}$ and the chain of intermediate invariants $\{C_i\}$. We recover the single-qubit rotations that stitch the segments together in two pieces. The \emph{interior} rotations $U$ sit between each invariant $C_{i-1}$ and the next basis gate $G_i$, and are responsible for matching the next invariant: the product $\mathrm{CAN}(C_{i-1})\, U\, G_i$ must be locally equivalent to $\mathrm{CAN}(C_i)$. The \emph{exterior} rotations $K$ then lift this local equivalence to genuine unitary equivalence, and are recovered through Cartan KAK as a separate post-processing step. Both $U$ and $K$ live in $\mathrm{PU}(2)\times\mathrm{PU}(2)$; we drop the segment index on $K$ for brevity.

For the three-segment case (the trivial initial segment $C_1 = G_1$ is omitted), the LP yields the chain
\[
C_2 \simeq G_2\, U_1\, G_1, \qquad T \simeq G_3\, U_2\, C_2.
\]
Each segment is solved for its $U_i$ independently. Replacing $C_2$ with its canonical form $\mathrm{CAN}(C_2)$ and recovering the surrounding $K$ rotations via KAK, the full unitary collapses to the alternating form
\[
T = K\, G_3\, K\, G_2\, K\, G_1\, K.
\]

For an $n$-segment LP solution, synthesis reduces to independently solving, for each segment $i$,
\begin{equation}
\mathrm{CAN}(C_i) \;\simeq\; \mathrm{CAN}(C_{i-1})\, U_{i-1}\, G_i,
\end{equation}
the effective saturation problem (Problem 12.2 of~\cite{peterson2020fixed}). After computing each interior rotation $U_i$ but before applying KAK to each intermediate $\mathrm{CAN}(C_i)$, the target takes the nested form:
\begin{equation}
T = K\Bigl( G_n\, U_{n-1}\, \Bigl( \cdots \Bigl( K\, G_2\, U_1\, G_1  K\Bigr) \cdots \Bigr)\Bigr) K.
\end{equation}

Segments could in principle be solved sequentially, substituting each $U_i$ back into subsequent templates to avoid nested KAK recoveries. We solve the segments independently. Post-hoc single-qubit merging is cheap, and the subproblems parallelize naturally without accumulating error.


The interior rotation is an element of \(\mathrm{PU}(2)\times\mathrm{PU}(2)\), six real parameters split across the two qubits. We benchmarked three coordinate choices in the restart loop (Euler angles, rotation vectors, unit quaternions); the production solver uses unit quaternions, which were empirically fastest. We use rotation-vector (RV) form below, with $\vec{v}_1, \vec{v}_2 \in \mathbb{R}^3$ the rotation vectors on each qubit, giving the ansatz
\begin{equation}
\label{eq:ansatz_final}
\begin{tikzcd}
& \gate[2]{C_{i-1}} & \gate{R(\vec{v}_1)} & \gate[2]{G_i} & \qw \\
& \qw               & \gate{R(\vec{v}_2)} & \qw         & \qw
\end{tikzcd}
\simeq
\begin{tikzcd}
& \gate[2]{C_i} & \qw \\
& \qw         & \qw
\end{tikzcd}
\end{equation}
where the rotation is defined as
\begin{equation}
\label{eq:RV_final}
R(\vec{v}) = \exp\Bigl(-i\,\frac{\vec{v}\cdot\sigma}{2}\Bigr),
\end{equation}
with \(\sigma=(\sigma_x,\sigma_y,\sigma_z)\) the Pauli matrices. 

We compare candidate and target through their Makhlin invariants $M[G]$, defining the residual
\begin{equation}
\label{eq:residual_final}
r(\vec{v}_1,\vec{v}_2) = M\Bigl[U(\vec{v}_1,\vec{v}_2)\Bigr] - M\bigl[U_{\mathrm{target}}\bigr],
\end{equation}
and minimize the least-squares cost
\begin{equation}
\label{eq:cost_final}
\min_{\vec{v}_1,\vec{v}_2\in\mathbb{R}^3} \; \|r(\vec{v}_1,\vec{v}_2)\|^2.
\end{equation}
An optimal solution with \(r=0\) is guaranteed to exist.

The Weyl objective requires computing eigenvalues in the hot path while the Makhlin objective does not, so we solve in two stages. Stage 1 drives the Makhlin residual down with Gauss--Newton and random restarts. Stage 1 alone is enough to certify local equivalence in principle, but in the iterative regime the Makhlin residual saturates a few orders of magnitude before the corresponding Weyl coordinates do, leaving the segment short of the near-exact compilation we want for the KAK stitch. Stage 2 polishes the best Stage-1 restart against a Weyl-distance objective with Levenberg--Marquardt, converging in a small handful of iterations at generic targets.

Given the optimal $\vec{v}_1^*, \vec{v}_2^*$, the interior rotation $U$ is built from the corresponding single-qubit rotations; KAK then recovers the exterior rotations to yield the segment unitary:
\begin{equation}
\mathrm{CAN}(C_i) = K\, \Bigl(G_i\, \bigl(R(\vec{v}_1^*) \otimes R(\vec{v}_2^*)\bigr)\, C_{i-1}\Bigr)\, K.
\end{equation}

Applying this procedure independently to each segment reduces general two-qubit synthesis to a series of low-dimensional, robust optimizations.

\subsection{$\rho$-reflections}
\label{sec:rho_reflection}

The QLR inequalities live in the $\mathrm{SU}(4)$ Weyl chamber, which is twice the size of the physical $\mathrm{PU}(4)$ chamber. Each $\mathrm{PU}(4)$ point therefore has two $\mathrm{SU}(4)$ representatives, related by the \emph{$\rho$-reflection}
\begin{equation}
\rho:\; (a_1,\, a_2,\, a_3) \;\longmapsto\; (\pi/2 - a_1,\; a_2,\; -a_3).
\end{equation}
$T$ and $\rho(T)$ are the same physical gate at different $\mathrm{SU}(4)$ chamber points, and the QLR construction sees only one of them. The two orientations are related by the circuit identity
\begin{equation}
\label{eq:rho_circuit_identity}
\begin{tikzcd}[column sep=tiny]
& \gate[2]{\scriptstyle\mathrm{CAN}(a_1, a_2, a_3)} & \qw \\
& \qw                                    & \qw
\end{tikzcd}
\;=\;
\begin{tikzcd}[column sep=tiny]
& \gate{Y} & \gate[2]{\scriptstyle\mathrm{CAN}(\frac{\pi}{2}-a_1, a_2, -a_3)} & \gate{Z} & \qw \\
& \qw      & \qw                                              & \gate{X} & \qw
\end{tikzcd}
\end{equation}
which converts $\rho(T)$ to $T$ using only single-qubit Paulis that absorb into the adjacent local-correction layers $K$, so neither the basis-gate sequence nor its cost is affected. Both stages of GULPS handle the $\rho$ ambiguity, but they see it from opposite sides.

The LP works on the $\mathrm{SU}(4)$ representative, so neither orientation can be inferred from the other (Corollary~25,~\cite{peterson2020fixed}). A fixed gate sequence may admit one of $\{T, \rho(T)\}$ inside its monodromy polytope while the other lies outside, so both must be tested explicitly. In the discrete formulation this is two sequential solves; in the continuous MILP it is a single binary variable selecting between $T$ and $\rho(T)$.

The numerics see the opposite degeneracy. Makhlin invariants are insensitive to $\rho$, so $T$ and $\rho(T)$ produce identical Stage-1 residuals, and a Gauss--Newton restart can converge to either basin regardless of which orientation the LP committed to. Before invoking the Stage-2 Weyl polish, GULPS evaluates the candidate's canonical Weyl coordinates; if they sit in the basin opposite to the LP's choice, the circuit identity~\eqref{eq:rho_circuit_identity} flips them back. The polish then drives toward the LP-aligned orientation, which is what the downstream stitch step expects when chaining intermediate invariants.

\section{Evaluation}
\label{sec:evaluation}

GULPS minimizes circuit cost by construction through the cost-ordered LP, so we focus the evaluation on wall-clock and fidelity. No prior tool minimizes cost over a heterogeneous discrete ISA, so we compare against existing decomposers in the regimes each one handles: NuOp~\cite{lao_designing_2021} and BQSKit~\cite{younis_berkeley_2021,davis2020towards} on iSWAP-family ISAs (Sec.~\ref{sec:general_comparison}), and Qiskit's \texttt{XXDecomposer}~\cite{peterson2022optimal} on $XX$-family ISAs, the only multi-gate heterogeneous case Qiskit handles (Sec.~\ref{sec:xx_comparison}). GULPS accepts any ISA of two-qubit unitaries; random ISAs do not appear here because no baseline accepts them.


\begin{figure}[t]
    \centering
    \includegraphics[width=\columnwidth]{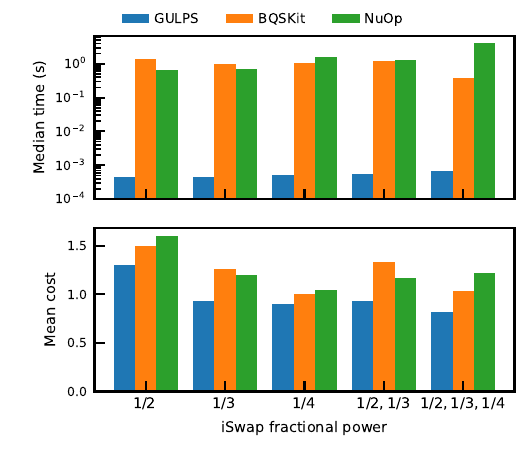}
    \caption{Median per-target wall-clock (top) and mean solution cost (bottom) across five iSWAP-power ISAs of increasing richness. 200 Haar-random targets per ISA.}
    \label{fig:cross_comparison}
\end{figure}

\subsection{Experimental Setup}
\label{sec:setup}

All experiments use GULPS v0.3.5 and standardized configurations of the comparison tools. NuOp enumerates depth-$k$ ans\"atze and runs one-shot continuous optimization over $6k$ parameters (U3 rotations between basis gates). BQSKit (v1.2.1) performs $A^\ast$-style search over partial circuit prefixes, optimizing each appended gate's parameters numerically; the default two-qubit engine is QSearch~\cite{davis2020towards}. NuOp targets depth-greedy length and BQSKit targets raw two-qubit gate count, so neither minimizes weighted cost over a heterogeneous ISA. Qiskit's \texttt{XXDecomposer} (Qiskit 2.3.1) is restricted to the $XX$-family (CX/RZX) and uses closed-form rules known for fixed XX interactions~\cite{peterson2022optimal}. All two-qubit targets are Haar-random unitaries~\cite{mele2024introduction} drawn with explicit seeds. Wall-clock times are measured on a single AMD Ryzen 5 5600X core.

\subsection{General ISAs: GULPS vs.\ NuOp and BQSKit}
\label{sec:general_comparison}

We compare on five iSWAP-power ISAs spanning single-gate to heterogeneous: $\{\sqrt[k]{\text{iSWAP}}\}$ for $k \in \{2, 3, 4\}$, $\{\sqrt{\text{iSWAP}}, \sqrt[3]{\text{iSWAP}}\}$, and $\{\sqrt{\text{iSWAP}}, \sqrt[3]{\text{iSWAP}}, \sqrt[4]{\text{iSWAP}}\}$. Each decomposer receives the same Haar-random targets and the same per-ISA cost weights (200 per ISA).

\textbf{On every ISA tested, GULPS is more than $500{\times}$ faster than BQSKit and more than $2500{\times}$ faster than NuOp at strictly lower circuit cost} (Fig.~\ref{fig:cross_comparison}). Median GULPS wall-clock stays under $0.6\,$ms across all five ISAs. BQSKit takes $0.22$--$0.37\,$s ($500$--$1800{\times}$ slower); NuOp takes $0.51$--$3.2\,$s ($2500$--$10{,}000{\times}$ slower). NuOp degrades sharply because its layer-by-layer sweep enumerates each depth from scratch.

GULPS finds the cost-optimal sentence on every target. Neither baseline does, and which one is closer to optimum depends on the ISA. When a baseline misses, it has selected an over-expressive sentence with strictly more two-qubit gates than the optimum. Mean solution cost lands $9$--$24\%$ above the GULPS optimum, and this per-block penalty accumulates across a full circuit (Sec.~\ref{sec:full_circuit}).

GULPS reaches the double-precision unitary-infidelity floor at ${\sim}10^{-15}$; BQSKit sits just above it. NuOp's reported infidelity plateaus around ${\sim}10^{-10}$, reflecting its configured tolerance and the joint-ansatz failure mode probed in Fig.~\ref{fig:depth_convergence_failure}.

\subsection{XX-Family ISAs: GULPS vs.\ XXDecomposer}
\label{sec:xx_comparison}

We measure total wall-clock for $1000$ Haar-random two-qubit targets across nine progressively richer CX-power ISAs, from $\{\text{CX}\}$ alone through $\{\text{CX}, \sqrt{\text{CX}}, \sqrt[3]{\text{CX}}, \sqrt[4]{\text{CX}}\}$.

GULPS is $3.9$--$9.2{\times}$ faster than \texttt{XXDecomposer} on every CX-power ISA tested (Fig.~\ref{fig:gulps_vs_xx}), at zero cost penalty. The speedup grows with ISA richness, from $3.9{\times}$ on $\{\sqrt[3]{\text{CX}}\}$ alone to $9.2{\times}$ on $\{\sqrt{\text{CX}}, \sqrt[3]{\text{CX}}, \sqrt[4]{\text{CX}}\}$. Output cost is identical on every ISA. The widening reflects a structural difference. GULPS' LP enumerates cost-ordered sentences in microseconds regardless of which gates appear, while \texttt{XXDecomposer} traverses $XX$-specific commutation chains whose per-block cost rises with ISA richness.

\begin{figure}[t]
    \centering
    \includegraphics[width=\columnwidth]{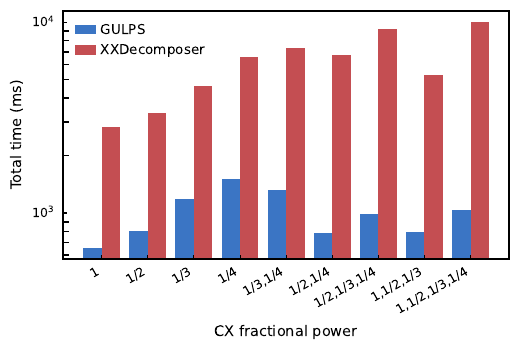}
    \caption{Total wall-clock for $1000$ Haar-random two-qubit decompositions across CX-power ISAs. ISA labels show the CX powers in the gate set; e.g., $\{1, 1/2, 1/3\} = \{\text{CX}, \sqrt{\text{CX}}, \sqrt[3]{\text{CX}}\}$.}
    \label{fig:gulps_vs_xx}
\end{figure}

\subsection{Full-Circuit Transpilation}
\label{sec:full_circuit}

Section~\ref{sec:xx_comparison} measures per-target decomposition cost; full-circuit transpilation aggregates this across many two-qubit blocks. We test on four benchmark families: QFT (structured algorithms such as Shor and phase estimation), EfficientSU2 with reps=$3$ and circular entanglement (hardware-efficient ansatz for VQE/QAOA), Quantum Volume at depth $n$, and Random at depth $4n$ with max-operands $2$. Qubit counts $n \in \{4, 8, 16, 32, 64\}$ on an all-to-all discrete RZZ basis $\{\pi/2, \pi/4, \pi/6\}$. Both pipelines run as Qiskit \texttt{PassManager}s (Unroll3qOrMore + Collect2qBlocks + ConsolidateBlocks + UnitarySynthesis or the GULPS translation plugin); we report the median of $3$ runs.

GULPS transpiles full circuits $7$--$19{\times}$ faster than the \texttt{XXDecomposer}-based pipeline (Fig.~\ref{fig:circuit_scaling}), with output circuit durations identical to within $1.2\,\mu$s at every size. The largest speedup appears on EfficientSU2 ($18.8{\times}$ at $32$ qubits, $15.6{\times}$ at $64$). Its many small two-qubit blocks are dominated by per-block transpiler overhead. For scale, a ${\sim}11{,}000$-gate Random circuit at $64$ qubits transpiles in $1.8\,$s under GULPS where \texttt{XXDecomposer} takes $20.2\,$s.

Across the full experiment suite ($>5000$ decompositions), GULPS reaches a unitary infidelity floor of ${\sim}10^{-15}$, limited by double-precision arithmetic.

\begin{figure}[t]
    \centering
    \includegraphics[width=\columnwidth]{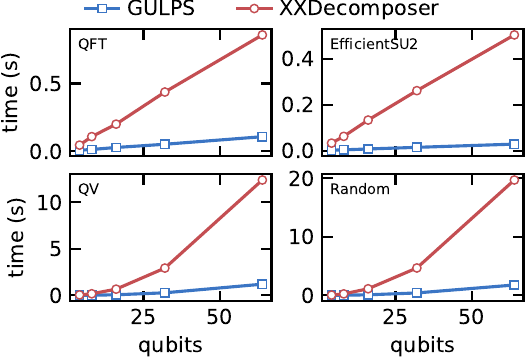}
    \caption{Median transpilation wall-clock time vs.\ qubit count for four benchmark circuit families on a discrete RZZ$\{\pi/2, \pi/4, \pi/6\}$ ISA with all-to-all connectivity. Median of $3$ runs.}
    \label{fig:circuit_scaling}
\end{figure}
\section{Discussion}
\label{sec:discussion}

The bulk runtime decomposes cleanly into two terms,
\[
T_{\text{total}} \;\approx\; q\,t_{\text{LP}} \;+\; n\,t_{\text{seg}},
\]
where $q$ is the queue depth (the number of sentences tried before the first feasible one), $t_{\text{LP}}$ is the per-LP cost, $n$ is the chosen sentence depth, and $t_{\text{seg}}$ is the per-segment numerical cost. Both $q$ and $n$ are Haar-volume-weighted properties of the ISA's coverage polytope and are computable analytically from the QLR inequalities; only $t_{\text{LP}}$ and $t_{\text{seg}}$ require empirical timing. To isolate the LP term, we sweep the progressive iSWAP-family $\mathcal{G}_k = \{\text{iSWAP}^{1/2}, \dots, \text{iSWAP}^{1/(k+1)}\}$ for $k = 1, \dots, 5$ over $200$ Haar-random targets per ISA. This family was chosen as an illustrative case where $(q, n)$ separate cleanly. Each smaller fractional power adds cheaper-but-narrower coverage, padding the LP queue with infeasible candidates while the expected depth of the optimal sentence barely grows. Other ISAs would mix $q$ and $n$ less cleanly, but the same $(q, n)$ decomposition applies. Across the sweep, $q$ grows $33{\times}$ ($1.2 \to 39.7$) while $n$ moves only $2.2 \to 3.3$. The empirical breakdown (Fig.~\ref{fig:profiling}) tracks this analytical pair, with numerics stable at $0.04$--$0.08\,$ms while LP grows from $32\,\mu$s to $527\,\mu$s. \textbf{$(q, n)$ are the two primary drivers of runtime on any new ISA, but they depend on the cost weights as much as on the gate set, since cost ordering determines which sentences are most likely to be optimal for the supplied workload.}

On the depth-$6$ $\sqrt[4]{\text{iSWAP}}$ apex target of Fig.~\ref{fig:depth_convergence_failure}, where naive joint optimization reaches $0/32$, GULPS succeeds in $1.4\,$ms with a length-$6$ sentence at machine-precision fidelity. Across the full apex sweep ($d=2,\dots,6$) it succeeds at every depth with wall-clock between $0.3$ and $3.8\,$ms.

The two terms have different ceilings. The LP term is a soft bottleneck: for batched compilation, the \texttt{monodromy} package precomputes the union of cost-ordered sentence polytopes into a coverage map that is scanned per target, eliminating per-target LP solves and amortizing this cost to near zero on a fixed ISA. The numerical term is the dominant bottleneck. Eliminating the Gauss--Newton/Levenberg--Marquardt stage would require an algebraic construction that produces segment solutions directly from the target invariants. We investigated this possibility and found no construction that is both computationally efficient and stable across the full Weyl chamber. The obstruction arises at degenerate Weyl faces, where the solution set ceases to be locally isolated and no globally consistent section exists. We therefore retain per-segment numerical optimization as the production solution.

The runtime distribution also exhibits a fat tail. Across the iSWAP- and CX-power benchmarks at $N{=}1000$ Haar targets per ISA, $p_{95}$ ranges from $1.2$ to $2.2\,$ms while the maximum observed sits between $2.4$ and $21.7\,$ms; the heaviest tail occurs on the single-gate $\sqrt[4]{\text{iSWAP}}$ ISA. The cause is geometric rather than algorithmic. Tail events correlate with targets near the Weyl-chamber symmetry faces ($c_1 \approx c_2$ or $c_2 \approx c_3$), where the Makhlin Jacobian becomes rank-deficient and the Gauss--Newton stage requires more restarts to converge. At a generic target, the local rotations that solve a segment form a discrete set of isolated points. At a symmetry face the canonical gate is fixed by a continuous one-parameter family of local rotations, so the solution set thickens into a curve. The Makhlin Jacobian drops a rank along that curve, and the Gauss--Newton step has no preferred direction.

\texttt{XXDecomposer} is tailored to the $XX$-family while GULPS handles arbitrary ISAs. The two pipelines have comparable complexity on the discrete problem; the observed speed gap reflects engineering effort that \texttt{XXDecomposer} could close with comparable optimization. \textbf{A completely generic method is already competitive with highly specialized tooling.} In the continuous one-parameter regime ($\theta$-free $R_{XX}$ or $\mathrm{CR}_Z$), Qiskit collapses synthesis to a trivial single-parameter optimization and is much faster. The MILP formulation of Section~\ref{sec:methods} extends GULPS to this regime for arbitrary base gates, and Section~\ref{sec:continuous} uses it to bound the cost achievable by any discrete sampling of $G^k$.

\begin{figure}[t]
    \centering
    \includegraphics[width=\columnwidth]{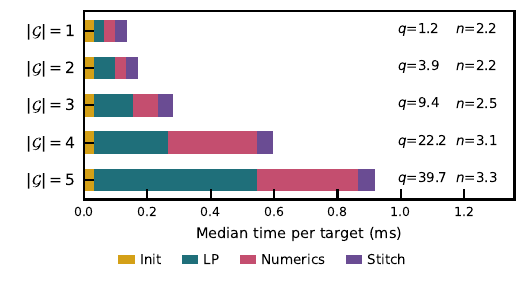}
    \caption{Median per-target time breakdown for the progressive iSWAP-family $\mathcal{G}_k$. Stacked bars decompose per-target wall-clock into initialization, LP enumeration, per-segment numerical synthesis, and stitching. The right-hand column reports the analytically-computed $q$ and $n$. $200$ Haar-random targets per ISA.}
    \label{fig:profiling}
\end{figure}

GULPS assumes block-consolidated two-qubit input, duration-based costs with linear composition, and exact unitary equivalence. The ISA must cover the Weyl chamber; otherwise the LP is infeasible. Noise-aware objectives, cross-talk, and approximate synthesis require different formulations. A complementary multi-qubit MILP approach~\cite{nagarajan2025provably} targets arbitrary qubit counts over fixed discrete gate sets such as Clifford+$T$, and therefore operates in an orthogonal regime. Extending GULPS beyond this setting requires either basis-specific analytical constructions or alternative objective functions; no closed-form construction is known for the generic XX+YY+ZZ CAN decomposition.

\section{Continuous Gate Families}
\label{sec:continuous}

\begin{figure}[t]
    \centering
    \includegraphics[width=\columnwidth]{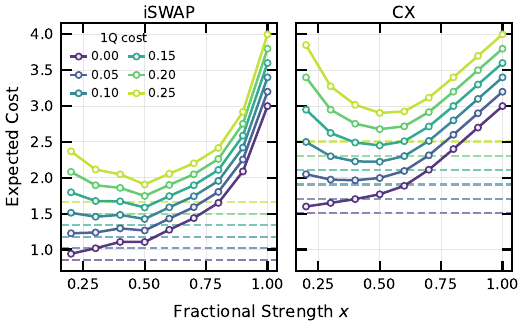}
    \caption{Expected cost vs.\ fractional strength for the iSWAP family (left) and the CX family (right) at several values of $\texttt{single\_qubit\_cost}$. Dashed lines mark the continuous-ISA baseline at the matching overhead. Discrete cost is computed by exact Haar integration over the monodromy coverage polytope; the continuous baseline is a sample estimate at 1000 targets.}
    \label{fig:cost_vs_power}
\end{figure}

The continuous-ISA formulation of Section~\ref{sec:methods} treats the gate strength $k \in [0,1]$ as a free decision variable. For a one-parameter family $G^k$, this gives the expected cost of an unconstrained continuous calibration as a baseline against which discrete samples of strengths can be measured.

The MILP of Section~\ref{sec:methods} is base-gate-agnostic: only the base monodromy $B$ enters as parameter data, while the QLR inequalities, semicontinuous strength variables, and monotonic-ordering constraints are assembled identically across families. The continuous backend is unoptimized and runs roughly $30{\times}$ slower per target than the discrete pipeline; we use it as a reference baseline, not as a production solver.

For each base gate we sweep the fractional strength $x$ in $G^x$ and compute the Haar-weighted expected cost of the resulting one-gate ISA by exact integration over the monodromy coverage polytope (Fig.~\ref{fig:cost_vs_power}). At zero local-rotation cost a smaller $x$ is always cheaper: e.g., the one-gate ISA $\{G^{0.25}\}$ is a strict superset of $\{G^{0.5}\}$, since each weaker gate simulates any stronger one by repetition at the same total 2Q weight. With nonzero overhead, the depth term dominates at small $x$ and the cost curves bend back upward; a minimum appears near $x = 1/2$. At every overhead, iSWAP-family ISAs have lower absolute cost than the corresponding CX-family ISAs, consistent with the known efficiency of iSWAP relative to CX at equivalent normalized duration.

Adding a second calibrated strength reduces the dispersion. We sweep all pairs $(x_1, x_2)$ on a $9{\times}9$ grid (Fig.~\ref{fig:synergy}). Single-gate costs span a $21\%$--$189\%$ gap relative to the continuous baseline as $x$ varies from $0.5$ to $1.0$; pairing each $x$ with its best grid complement narrows the range to $8\%$--$21\%$. The optimal pair on iSWAP is $(0.30,\, 0.50)$, with $\sqrt{\text{iSWAP}}$ as the primary calibration and $x \approx 0.3$ as the complement. All percentages in this paragraph and Table~\ref{tab:continuous_summary} use $\texttt{single\_qubit\_cost}{=}0.1$; Fig.~\ref{fig:cost_vs_power} shows the single-strength dependence on overhead.

\begin{figure}[t]
    \centering
    \includegraphics[width=.95\columnwidth]{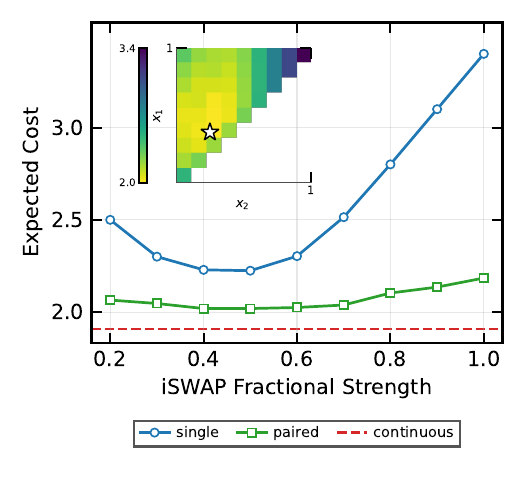}
    \caption{Two-strength synergy on the iSWAP$^x$ family at $\texttt{single\_qubit\_cost}{=}0.1$. Single-gate cost vs.\ $x$ alone (\emph{single}) and when $x$ is joined by its best complement from the $9{\times}9$ pair grid (\emph{paired}); the continuous-ISA baseline is the asymptote. Inset: Haar-weighted expected cost over the pair grid; the star marks the optimal pair $(0.30,\, 0.50)$.}
    \label{fig:synergy}
\end{figure}

The same scan on the CX, iSWAP, and SYC ($\text{fSim}(\pi/2, \pi/6)$) families gives Table~\ref{tab:continuous_summary}. The best single strength leaves a $17\%$--$24\%$ gap to the continuous baseline; a second strength closes that to $6\%$--$9\%$, a ${\sim}2.5{\times}$ reduction in each case, with optimal strengths and absolute costs varying by family.

\begin{table}[t]
\centering
\caption{Best single and best pair on three base-gate families at $\texttt{single\_qubit\_cost}{=}0.1$.}
\label{tab:continuous_summary}
\begin{tabular}{lcccc}
\toprule
Base gate & Best $x^\star$ & Single gap & Best pair $(x_1^\star, x_2^\star)$ & Pair gap \\
\midrule
CX    & $0.50$ & $+16.6\%$ & $(0.40,\, 0.50)$ & $+5.9\%$ \\
iSWAP & $0.50$ & $+21.2\%$ & $(0.30,\, 0.50)$ & $+8.3\%$ \\
SYC   & $0.30$ & $+24.0\%$ & $(0.30,\, 0.50)$ & $+8.8\%$ \\
\bottomrule
\end{tabular}
\end{table}

The continuous-ISA cost is a lower bound on the Haar-averaged expected cost achievable by any discrete sampling of $G^k$ at fixed local-rotation overhead. On Haar-random workloads, two strengths approach this bound to within single-digit percent and a third yields diminishing improvement. The optimal strengths and the absolute gap shift with the local-rotation overhead, since the overhead sets the depth penalty for weaker gates. Whether the second calibration is worth its hardware cost is platform-specific; the scan above runs on any user-supplied base gate through the public package.

\section{Conclusion}
GULPS extracts reachability and intermediate invariants from monodromy geometry, decomposing two-qubit synthesis into segment-wise local-rotation fits gated by a linear program. It compiles against arbitrary heterogeneous ISAs at speeds competitive with closed-form decomposers for individual gate families, recovers the cost-optimal sentence on every Haar-random target tested, and reaches the double-precision unitary-infidelity floor across more than $5000$ decompositions. The continuous-ISA formulation provides a Haar-averaged lower bound on expected cost, which two-strength discrete calibrations approach to within $6\%$--$9\%$ on the base gates we tested.


\section*{Acknowledgments}
We thank Ali Javadi-Abhari and Eric C. Peterson for their discussions and feedback during the development of this work. 

\appendices

\section{Software Availability}
\label{sec:software}

GULPS is distributed on PyPI as \texttt{gulps}, is part of the Qiskit ecosystem, and registers as a Qiskit translation-stage plugin. The continuous-ISA solver of Section~\ref{sec:continuous} requires the optional CPLEX backend, available via \texttt{pip install gulps[cplex]}. Source, tests, and tutorial notebooks are available at \url{https://github.com/evmckinney9/gulps}.

\section{Sentence ordering}
\label{app:sentence_ordering}

\begin{theorem}
\label{th:ordering}
For any circuit sentence \(S\) and for any permutation \(\pi\) of the fixed gates \((G_1,\dots, G_n)\), there exist local unitaries such that the reordered sentence
\[
\tilde{S} : \theta \mapsto \Bigl( \tilde{K}_0(\theta),\, G_{\pi(1)},\, \tilde{K}_1(\theta),\, G_{\pi(2)},\, \dots,\, G_{\pi(n)},\, \tilde{K}_n(\theta) \Bigr)
\]
implements a unitary that is locally equivalent to that produced by \(S\). In other words, the circuit polytope is invariant under any permutation of the fixed gates.
\end{theorem}

\begin{proof}
Any permutation is a product of adjacent swaps, so it suffices to show that exchanging one adjacent pair $(G_i, G_{i+1})$ leaves the polytope unchanged. The general adjacent swap reduces to the depth-2 case: if a two-gate sentence $G_{i+1}\, K_i\, G_i$ is locally equivalent to $G_i\, K'_i\, G_{i+1}$, then by definition there exist locals $L, L'$ with
\[
G_{i+1}\, K_i\, G_i \;=\; L\,\bigl(G_i\, K'_i\, G_{i+1}\bigr)\, L',
\]
and $L, L'$ can be absorbed into the surrounding layers $K_{i+1}, K_{i-1}$ of any longer sentence without disturbing any other fixed gate. So the only thing to prove is the depth-2 identity $\mathrm{Pol}(G_1, G_2) = \mathrm{Pol}(G_2, G_1)$.

This is now visible directly from the polytope inequalities at Eq.~\eqref{eq:qlr_inequalities}. Each inequality is linear in the alcove coordinates of $G_1$, $G_2$, and the composition; the coordinates $\alpha$ of $G_1$ and $\beta$ of $G_2$ enter through identical-shape summands. The list of inequalities is itself symmetric under exchanging the two input slots for every QLR datum $(a,b,c,d)$ that selects a constraint, the swapped datum $(b,a,c,d)$ selects another. Simultaneously relabeling $(\alpha, a) \leftrightarrow (\beta, b)$ therefore leaves the entire system unchanged.
\end{proof}

\bibliographystyle{IEEEtran}
\bibliography{refs, zotero_refs}

\end{document}